\documentclass{eptcs}
\usepackage[ddmmyy,hhmmss]{datetime}
\usepackage{amsthm}
\usepackage{amssymb}
\usepackage{graphics,latexsym,amsfonts}
\usepackage{mathtools}
\usepackage{epigraph}
\usepackage{picture}
\usepackage{enumitem}
\setlist[enumerate]{leftmargin=*}
\usepackage{color}
\usepackage{xspace}
\usepackage{setspace} 

\newtheorem{theorem}{Theorem} 
\newtheorem{definition}{Definition}
\newtheorem{lemma}{Lemma}

\newtheorem{example}{Example}

\newcommand{\NP}{NP}
\newcommand{\NEXP}{NEXPTIME}
\newcommand{\terms}{\mathcal{T}_{\varphi}}

\def\titlerunning{The Satisfiability Problem for Choice Correspondences}
\def\authorrunning{D.\ Cantone, A.\ Giarlotta \& S.\ Watson}

\def\eod {{\unskip\nobreak\hfil\penalty50
\hskip2em\hbox{}\nobreak\hfil \rule{2.57mm}{2,57mm}\hspace{.71pt}
\parfillskip=0pt \finalhyphendemerits=0 \par \medskip}}

\author{
Domenico Cantone\thanks{Department of Mathematics and Computer Science, University of Catania, Italy. Email: cantone@dmi.unict.it}$\:$,
Alfio Giarlotta\thanks{Department of Economics and Business, University of Catania, Italy. Email: giarlott@unict.it}$\:$,
Stephen Watson\thanks{Department of Mathematics and Statistics, York University, Toronto, Canada. Email: watson@mathstat.yorku.ca}}

\title{\bf The Satisfiability Problem for Boolean Set Theory\\ with a Choice Correspondence}

\newcommand{\defAs}{\coloneqq}

\newcommand{\pow}{\mathrm{Pow}}

\newcommand{\myprec}{\mathrel{\leqslant\hspace{-5.8pt}\raisebox{1pt}{$\cdot$}}}

\newcommand{\BSTC}{\mathsf{BSTC}}

\newcommand{\model}{\ensuremath{\mbox{\boldmath $\mathcal{M}$}}\xspace}
\newcommand{\M}[1]{#1^{\scriptscriptstyle M}}

\newcommand{\Mbeta}[1]{#1^{\scriptscriptstyle M_{\beta}}}
\newcommand{\Malpha}[1]{#1^{\scriptscriptstyle M_{\alpha}}}

\newcommand{\choice}{\mathtt{c}}

\newcommand{\true}{\ensuremath{\mbox{$\mathsf{1}$}}\xspace}
\newcommand{\false}{\ensuremath{\mbox{$\mathsf{0}$}}\xspace}

\def\F{\mathcal{F}}

\begin{document}

\maketitle

\begin{abstract}
\noindent 
Given a set $U$ of alternatives, a choice (correspondence) on $U$ is a contractive map $c$ defined on a family $\Omega$ of nonempty subsets of $U$. 
Semantically, a choice $c$ associates to each menu $A \in \Omega$ a nonempty subset $c(A) \subseteq A$ comprising all elements of $A$ that are deemed selectable by an agent.
A choice on $U$ is total if its domain is the powerset of $U$ minus the empty set, and partial otherwise.
According to the theory of revealed preferences, a choice is rationalizable if it can be retrieved from a binary relation on $U$ by taking all maximal elements of each menu.   
It is well-known that rationalizable choices are characterized by the satisfaction of suitable axioms of consistency, which codify logical rules of selection within menus. 
For instance, \textsf{WARP} (Weak Axiom of Revealed Preference) characterizes choices rationalizable by a transitive relation.
Here we study the satisfiability problem for unquantified formulae of an elementary fragment of set theory involving a choice function symbol $\choice$, the Boolean set operators and the singleton, the equality and inclusion predicates, and the propositional connectives. 
In particular, we consider the cases in which the interpretation of $\choice$ satisfies any combination of two specific axioms of consistency, whose conjunction is equivalent to \textsf{WARP}. 
In two cases we prove that the related satisfiability problem is \NP-complete, whereas in the remaining cases we obtain \NP-completeness under the additional assumption that the number of choice terms is constant.

\medskip

\noindent \textbf{Keywords:} Decidability; NP-completeness; choice; axioms of choice consistency; \textsf{WARP}. 

%

\end{abstract}


\section{Introduction} \label{SECT:intro}

In this paper we examine the decidability of the satisfiability problem connected to \textit{rational choice theory}, which is a framework to model social and economic behavior.  
A choice on a set $U$ of alternatives is a correspondence $B \mapsto c(B)$ associating to ``feasible menus" $B \subseteq U$ nonempty ``choice sets" $c(B) \subseteq B$.
This choice can be either \textit{total} (or \textit{full}) -- i.e, defined for all nonempty subsets of the ground set $U$ of alternatives -- or \textit{partial} -- i.e., defined only for suitable subsets of $U$.

According to the \textit{Theory of Revealed Preferences} pioneered by the economist Paul Samuelson~\cite{Sam38}, 
preferences of consumers can be derived from their purchasing habits: in a nutshell, an agent's choice behavior is observed, and her underlying preference structure is inferred.
The preference revealed by a primitive choice is typically modeled by a binary relation on $U$.
The asymmetric part of this relation is informative of a ``strict revealed preference" of an item over another one, whereas its symmetric part codifies a ``revealed similarity" of items.
Then a choice is said to be \textit{rationalizable} when the observed behavior can be univocally retrieved by maximizing the relation of revealed preference.

Since the seminal paper of Samuelson, a lot of attention has been devoted to notions of rationality within the framework of choice theory: see, among the many contributions to the topic, the classical papers~\cite{Hou50,Arr59,Ric66,Han68,Sen71}.  
(See also the book~\cite{AleBouMon07} for the analysis of the links among the theories of choice, preference, and utility. For a very recent contribution witnessing the fervent research on the topic, see~\cite{ChaEchShm17}.) 
Classically, the rationality of an observed choice behavior is connected to the satisfaction of suitable \textit{axioms of choice consistency}: these are rules of selections of items within menus, codified by means of sentences of second-order monadic logic, universally quantified over menus.
Among the several axioms introduced in the specialized literature, let us recall the following:

$\bullet$ \emph{standard contraction consistency} $(\alpha)$, introduced by Chernoff~\cite{Che54};

$\bullet$ \textit{standard expansion consistency} $(\gamma)$, and \textit{binary expansion consistency} $(\beta)$, both due to Sen~\cite{Sen71};

$\bullet$ \textit{the weak axiom of revealed preference} (\textsf{WARP}), due to Samuelson~\cite{Sam38}.\\
It is well-known that, under suitable assumptions on the domain, a choice is rationalizable if and only if the two standard axioms of consistency  $(\alpha)$ and $(\gamma)$ hold. 
Further, the rationalizing preference satisfies the property of transitivity if and only if axioms $(\alpha)$ and $(\beta)$ hold if and only if \textsf{WARP} holds: in this case, we speak of a \textit{transitively rationalizable} choice. 
Section~\ref{SECT:preliminaries} provides the background to choice theory.

Although the mathematical economics literature on the topic is quite large, there are no contributions which deal with related decision procedures in choice theory.
In this paper we start filling this gap.
Specifically, we study the satisfiability problem for unquantified formulae of an elementary fragment of set theory (denoted $\BSTC$) involving a choice function symbol $\choice$, the Boolean set operators $\cup$, $\cap$, $\setminus$ and the singleton $\{\cdot\}$, the predicates equality $=$ and inclusion $\subseteq$, and the propositional connectives $\land$, $\lor$, $\neg$, $\implies$, etc. 
Here we consider the cases in which the interpretation of $\choice$ is subject to any combination of the axioms of consistency $(\alpha)$ and $(\beta)$, whose conjunction is equivalent to \textsf{WARP}. 
In two cases we prove that the related satisfiability problem is \NP-complete, whereas in the remaining cases we obtain \NP-completeness only under the additional assumption that the number of choice terms is constant. 

By depriving the $\BSTC$-language of the choice function symbol $\choice$, we obtain the fragment \textsf{2LSS} (here denoted $\BSTC^{-}$) whose decidability was known since the birth of \textit{Computable Set Theory} in the late 70's. 
In Section~\ref{appendixDecProc} we rediscover such result as a by-product of the solution to the satisfiability problem of $\BSTC$ under the \textsf{WARP}-semantics: the latter is based on a novel term-oriented non-clausal approach. 
The reader can find extensive information on Computable Set Theory in the monographs \cite{CanFerOmo89a,CanOmoPol01,SchCanOmo11,CanUrs17}.

For our purposes, it will be relevant to solve the following \emph{lifting problem}: Given a partial choice satisfying some axioms of consistency, can we suitably characterize whether it is extendable to a total choice satisfying the same axioms? 
The lifting problem for the various combinations of axioms $(\alpha)$ and $(\beta)$ is addressed in depth in Section~\ref{SECT:liftings}.
In particular, in the case of finite choice correspondences, our characterizations turn out to be effective and, with only one exception, expressible in the same $\BSTC$-language. 
This facilitates the design of effective procedures for the solution of the satisfiability problems of our concern. 
The syntax and semantics of the $\BSTC$-language, as well as the solutions of the satisfiability problem for $\BSTC$-formulae under the various combinations of axioms $(\alpha)$ and $(\beta)$ are presented in Section~\ref{SECT:satProb}. 
Finally, in Section~\ref{SECT:Conclusions}, we draw our conclusions and hint at future developments.


\section{Preliminaries on choice theory} \label{SECT:preliminaries}

\newcommand{\powPlus}{\pow^{\hbox{\tiny{+}}}(U)}

Hereafter, we fix a nonempty set $U$ (the ``universe").
Let $\pow(U)$ be the family of all subsets of $U$, and $\powPlus$ the subfamily $\pow(U) \setminus \{\emptyset\}$.
The next definition collects some basic notions in choice theory.

\newcommand{\cPlus}{c^{\hbox{\tiny{+}}}}
\newcommand{\oc}{\overline{c}}
\newcommand{\ocPrime}{\overline{\cPlus}}

\begin{definition} \rm \label{DEF:preliminary deff on choice}
    Let $\Omega \subseteq \powPlus$ be nonempty.
    A map $f \colon \Omega \to \pow(U)$ is \textit{contractive} if $f(B) \subseteq B$ for each $B \in \Omega$. 
    A \textit{choice correspondence} on $U$ is a contractive map that is never empty-valued, i.e.,
    $$
    c \colon \Omega \to \powPlus \quad \hbox{such that} \quad c(B) \subseteq B \quad \hbox{for each} \;\; B \in \Omega\,.
    $$
    In this paper, we denote a choice correspondence on $U$ by $c \colon \Omega \rightrightarrows U$, and simply refer to it as a \textit{choice}. 
    The family $\Omega$ is the \textit{choice domain} of $c$, sets in $\Omega$ are \textit{(feasible) menus}, and elements of a menu are \textit{items}. 
    Further, we say that $c \colon \Omega \rightrightarrows U$ is \textit{total} (or \textit{full}) if $\Omega = \powPlus$, and \textit{partial} otherwise.
    The \textit{rejection map} associated to $c$ is the contractive function $\oc \colon \Omega \to \pow(U)$ defined by $\oc(B) := B \setminus c(B)$ for all $B \in \Omega$. \eod
\end{definition}

Given a choice $c\colon \Omega \rightrightarrows U$, the \textit{choice set} $c(B)$ of a menu $B$ collects the elements of $B$ that are deemed selectable by an economic agent.
Thus, in case $c(B)$ contains more than one element, the selection of a single element of $B$ is deferred to a later time, usually with a different procedure (according to additional information or ``subjective randomization", e.g., flipping a coin). 
Notice that the rejection map associated to a choice may fail to be a choice, since the rejection set of some menu can be empty.

%

The next definition recalls the classical notion of a rationalizable choice.

\begin{definition} \rm \label{DEF:rationalizable choice}
    A choice $c \colon \Omega \rightrightarrows U$ is \textit{rationalizable} (or \textit{binary}) if there exists a binary relation $\precsim$ on $U$ such that the equality\footnote{Recall that $\max_{\precsim} B = \{a \in B : (\nexists b \in B) (a \prec b)\}$, where $a \prec b$ means $a \precsim b$ and $\neg(b \precsim a)$.} $c(B) = \max_{\precsim} B$ holds for all menus $B \in \Omega$. \eod
\end{definition}

The revealed preference theory approach postulates that preferences can be derived from choices. 
The preference revealed by a primitive choice is modeled by a suitable binary relation on the set of alternatives.
Then a choice is rationalizable whenever the observed behavior can be fully explained (i.e., retrieved) by constructing a binary relation of revealed preference.

The rationalizability of choice is traditionally connected to the satisfaction of suitable \textit{axioms of choice consistency}. 
These axioms codify rules of coherent behavior of an economic agent.
Among the several axioms that are considered in the literature, the following are relevant to our analysis (a universal quantification on all the involved menus is implicit):
\smallskip

\noindent 
\begin{tabular}{lll}
  \textbf{axiom $(\alpha)$} [\emph{standard contraction consistency}]: & 
      $A \subseteq B \;\; \Longrightarrow \;\; A \cap c(B) \subseteq c(A)$\\[.1cm]
  \textbf{axiom $(\gamma)$} [\emph{standard expansion consistency}]: &
      $c(A) \cap c(B) \subseteq c(A \cup B)$\\[.1cm]
  \textbf{axiom $(\beta)$} [\emph{symmetric expansion consistency}]: &
      $\big(A \subseteq B \: \wedge \: c(A) \cap c(B) \neq \emptyset \big) \;\; \Longrightarrow \;\; c(A) \subseteq c(B)$\\[.1cm]
  \textbf{axiom $(\rho)$} [\emph{standard replacement consistency}]: &
      $c(A) \setminus c(A \cup B) \neq \emptyset \;\; \Longrightarrow \;\; B \cap c(A \cup B) \neq \emptyset$\\[.1cm]
  \textbf{\textsf{WARP}} [\emph{weak axiom of revealed preference}]: &
      $\big(A \subseteq B \: \wedge \: A \cap c(B) \neq \emptyset \big) \;\; \Longrightarrow \;\; c(A) = A \cap c(B)$.    
\end{tabular}\\

Axiom $(\alpha)$ was studied by Chernoff~\cite{Che54}, whereas axioms $(\gamma)$ and $(\beta)$ are due to Sen~\cite{Sen71}.
\textsf{WARP} was introduced by Samuelson in~\cite{Sam38}.
Axiom $(\rho)$ has been recently introduced in~\cite{CanGiaGreWat16}, in connection to the transitive structure of the relation of revealed preference.

Upon reformulating these properties in terms of items, their semantics becomes clear.
Chernoff's axiom $(\alpha)$ states that any item selected from a menu $B$ is still selected from any submenu $A \subseteq B$ containing it. 
Sen's axiom $(\gamma)$ says that any item selected from two menus $A$ and $B$ is also selected from the menu $A \cup B$ (if feasible).
The expansion axiom $(\beta)$ can be equivalently written as follows: if $A \subseteq B$, $x,y \in c(A)$ and $y \in c(B)$, then $x \in c(B)$. 
In this form, $(\beta)$ says that if two items are selected from a menu $A$, then they are simultaneously either selected or rejected in any larger menu $B$.
Axiom $(\rho)$ can be equivalently written as follows: if $y \in c(B) \setminus c(B \cup \{x\})$, then $x \in c(B \cup \{x\})$. 
In this form, $(\rho)$ says that if an item $y$ is selected from a menu $B$ but not from the larger menu $B \cup \{x\}$, then the new item $x$ is selected from $B \cup \{x\}$. 
\textsf{WARP} summarizes features of contraction and expansion consistency in a single -- and rather strong, despite its name -- axiom, in fact it is equivalent to the conjunction of $(\alpha)$ and $(\beta)$~\cite{Sen71}. 

\section{Liftings}\label{SECT:liftings}

In this section we examine the ``lifting problem": this corresponds to finding necessary and sufficient conditions such that a partial choice satisfying some axioms of consistency can be extended to a total choice satisfying the same axioms. We shall exploit such conditions in the decision results to be presented in Section~\ref{SECT:satProb}.

The next definition makes the notion of lifting formal.
\newcommand\restrict[1]{\raisebox{-.5ex}{$|$}_{#1}}

\begin{definition} \rm \label{DEF:lifting}
    Let $c \colon \Omega \rightrightarrows U$ be a choice. 
    Given a nonempty set $\F$ of sentences of second-order monadic logic, we say that $c$ has the $\F$\textit{-lifting property} if there is a total choice $\cPlus \colon \powPlus \rightrightarrows U$ extending $c$ (i.e., $\cPlus \restrict{\Omega} = c$) and satisfying all formulae in $\F$.
    In this case, $\cPlus$ is called an $\F$\textit{-lifting} of $c$.
    (Of course, we are interested in cases such that $\F$ is a family of axioms of choice consistency.) 
    Whenever $\F$ is a single formula, we simplify notation and write, e.g., $(\alpha)$-lifting, $\textsf{WARP}$-lifting, etc.
    Similarly, we say that $c$ has the \textit{rational lifting property} if there is a total choice $\cPlus$ that is rational and extends $c$. 
\end{definition}

Notice that whenever $\F$ is a nonempty set of axioms of choice consistency (which are formulae in prenex normal form where all quantifiers are universal), if a choice has the $\F$-lifting property, then it automatically satisfies all axioms in $\F$. 
The same reasoning applies for the rational lifting property, since it is based on the existence of a binary relation of revealed preference that is fully informative of the choice.

On the other hand, it may happen that a partial choice satisfies some axioms in $\F$ but there is no lifting to a total choice satisfying the same axioms.
The next examples exhibit two instances of this kind.
(To simplify notation, we underline all items that are selected within a menu: for instance $\underline{x}\,y$ and $\underline{x}\,y\,\underline{z}$ stand for, respectively, $c(\{x,y\})=\{x\}$ and $c(\{x,y,z\})=\{x,z\}$. Obviously, we always have $\underline{x}$ for any $\{x\} \in \Omega$, so we can safely omit defining $c$ for singletons.)

\begin{example} \rm \label{EX:no lifting 1}
    Let $U= \{x,y,z\}$ and $\Omega =  \{B \subseteq U : 1 \leq \vert B \vert \leq 2\}$.
    Define a partial choice $c \colon \Omega \rightrightarrows U$ by $\underline{x}\,y$, $\underline{y}\,z$, and $x\,\underline{z}$. 
    This choice is rationalizable by the (cyclic) preference $\precsim$ defined by $x \prec y \prec z \prec x$. 
    However, $c$ does not admit any rational lifting to a total choice $\cPlus$, since we would have $\cPlus(U) = \max_{\precsim} U = \emptyset$.
\end{example}

\begin{example} \rm \label{EX:no lifting 2}
    Let $U= \{x,y,z,w\}$ and $\Omega = \powPlus \setminus \{\{x,w\},\{y,z\},\{y,w\},\{z,w\},U\}$.
    Define a partial choice $c \colon \Omega \rightrightarrows U$ by
    $
    \underline{x}\,\underline{y}\,,\; \underline{x}\,\underline{z}\,,\; x\,\underline{y}\,\underline{z}\,,\;
    \underline{x}\,y\,\underline{w}\,,\; \underline{x}\,\underline{z}\,w\,,\; \underline{y}\,z\,\underline{w}\,.
    $ 
    One can easily check that $c$ satisfies axiom $(\alpha)$ (but it fails to be rationalizable). 
    On the other hand, $c$ admits no $(\alpha)$-lifting, since $\cPlus(U) \neq \emptyset$ violates axiom $(\alpha)$ for any choice $\cPlus$ extending $c$ to the full menu $U$.  
\end{example}


\subsection{Lifting of axiom $(\alpha)$} \label{SECT:lifting alpha}

In this section we characterize the choices that are $(\alpha)$-liftable.
To that end, it is convenient to reformulate axiom $(\alpha)$ in terms of the monotonicity of the rejection map.
We need the following preliminary result, whose simple proof is omitted, and a technical definition.

\begin{lemma}\label{lemmaEquiv}
    Let $A \subseteq B \subseteq U$. 
    For any pair of sets $A',B' \subseteq U$, we have $\:A \cap B' \subseteq A' \;\Longleftrightarrow \; A \setminus A' \subseteq B \setminus B'\,.$
\end{lemma}

\begin{definition} \rm \label{DEF:relativization}      
    Let $c \colon \Omega \rightrightarrows U$ be a choice. 
    Given a menu $A \in \powPlus$, the \emph{relativized choice domain $\Omega_{A}$ w.r.t.\ $A$} is the collection of all submenus of $A$, that is, $B \in \Omega$ such that $B \subseteq A$; in symbols,    
    \begin{equation}\label{relativizedDomain}
        \Omega_{A} \defAs \{B \in \Omega : B \subseteq A\}\,.
    \end{equation}
    A set $\mathcal{B} \subseteq \Omega$ of menus is \emph{$\subseteq$-closed w.r.t.\ $\Omega$} if $B \in \mathcal{B}$, for every $B \in \Omega$ such that $B \subseteq \bigcup \mathcal{B}$.   
    \eod
\end{definition}


In view of Lemma~\ref{lemmaEquiv}, axiom $(\alpha)$ can be equivalently rewritten as follows: \begin{equation}\label{equivAlpha}
A \subseteq B \quad \Longrightarrow \quad \oc(A) \subseteq \oc(B)\,.
\end{equation}
In this form, axiom $(\alpha)$ just asserts that enlarging the set of alternatives may only cause the set of neglected members to grow. 
As announced, we have:

\begin{theorem} \label{THM:lifting alpha}
A partial choice $c \colon \Omega \rightrightarrows U$ has the $(\alpha)$-lifting property if and only if the following two conditions hold: 
\begin{enumerate}[label=\text{(\alph*)}, start=1]

\item\label{b} $A \subseteq B \;\; \Longrightarrow \;\; A \cap c(B) \subseteq c(A)$,~ for all $A,B \in \Omega\,$;

\item\label{c} $\bigcup \mathcal{B} \setminus \bigcup_{B \in \mathcal{B}} \overline{c}(B) \neq \emptyset$, for every $\emptyset \neq \mathcal{B} \subseteq \Omega$ such that $\mathcal{B}$ is $\subseteq$-closed w.r.t.\ $\Omega\,$.
\end{enumerate}
\end{theorem}

\begin{proof} 
For necessity, assume that $c \colon \Omega \rightrightarrows U$ can be extended to a total choice $\cPlus$ on $U$ satisfying axiom $(\alpha)$, and let $\oc$ be the associated rejection map of $c$.
Since $c = \cPlus\restrict{\Omega}$, condition~\ref{b} follows immediately from axiom $(\alpha)$ for $\cPlus$. 
To prove that $c$ satisfies condition~\ref{c} as well, let $\mathcal{B}$ be a nonempty $\subseteq$-closed subset of $\Omega$.
By the equivalent formulation (\ref{equivAlpha}) of axiom $(\alpha)$, we obtain $\ocPrime(B) \subseteq \ocPrime(\bigcup \mathcal{B})$ for every $B \in \mathcal{B}$, where $\ocPrime$ is the associated rejection map of $\oc$. Hence
\[  \bigcup_{B \in \mathcal{B}} \ocPrime(B) \; \subseteq \; \ocPrime\big(\bigcup \mathcal{B}\big) \; = \; \bigcup \mathcal{B} \setminus \cPlus\big(\bigcup \mathcal{B}\big)
\]
holds. 
It follows that
\[  \emptyset \; \neq \; \cPlus\big(\bigcup \mathcal{B}\big) \; \subseteq \; \bigcup \mathcal{B} \setminus \bigcup\nolimits_{B \in \mathcal{B}} \ocPrime(B) \; = \; \bigcup \mathcal{B} \setminus \bigcup\nolimits_{B \in \mathcal{B}} \oc(B)\,,
\]
thus showing that \ref{c} holds.
This completes the proof of necessity.

\smallskip

For sufficiency, assume that \ref{b} and \ref{c} hold for the choice $c \colon \Omega \rightrightarrows U$.
For each $A \in \powPlus$, define
\[    \cPlus(A) \: \defAs \: A \setminus \bigcup\nolimits_{B \in \Omega_{A}} \oc(B)\,,
\]
where we recall that $\Omega_{C}$ is the relativized choice domain w.r.t.\ to $C$ (cf.\ Definition~\ref{DEF:relativization}).
In what follows we prove that the map $\cPlus \colon \powPlus \to \pow(U)$ is a well-defined choice, which extends $c$ and satisfies axiom $(\alpha)$.

Since the map $\cPlus$ is obviously contractive by definition, to prove that it is a well-defined choice it suffices to show that it is never empty-valued.
Toward a contradiction, assume that $\cPlus(A) = \emptyset$ for some $A \in \powPlus$.
The definition of $\cPlus$ readily yields
$
  A \; = \; \bigcup\nolimits_{B \in \Omega_{A}} \oc(B) \; \subseteq \; \bigcup \Omega_{A} \; \subseteq \; A\,,
$
which implies $\Omega_{A} \neq \emptyset$ and $\bigcup_{B \in \Omega_{A}} \oc(B) = \bigcup \Omega_{A}$, since $\Omega_{A}$ is $\subseteq$-closed w.r.t.\ $\Omega$.
However, this contradicts \ref{c}. 
Thus, $\cPlus \colon \powPlus \rightrightarrows U$ is a well-defined (total) choice. 
 
Next, we show that $\cPlus$ extends $c$. 
Let $B \in \Omega$. 
Plainly, we have $B \in \Omega_{B}$ and $A \subseteq B$ for each $A \in \Omega_{B}$.
Property \ref{b} yields $A \cap c(B) \subseteq c(A)$, and so $\oc(A) \subseteq \oc(B)$ by Lemma~\ref{lemmaEquiv}. 
Thus, we obtain 
\[\cPlus(B) \: = \: B \setminus \bigcup\nolimits_{A \in \Omega_{B}} \oc(A) \: = \: B \setminus \oc(B) \: = \: c(B)\,,
\]
which proves the claim.

Finally, we show that $\cPlus$ satisfies $(\alpha)$. 
Let $\emptyset \neq B \subseteq C \subseteq U$. 
Since $B \subseteq C$ and $\Omega_{B} \subseteq \Omega_{C}$, we have:\\
\[
B \cap \cPlus(C)  = \; B \cap \left(C \setminus \bigcup\nolimits_{A \in \Omega_{C}} \oc(A)\right) = \; B \setminus \bigcup\nolimits_{A \in \Omega_{C}} \oc(A)
\subseteq \; B \setminus \bigcup\nolimits_{A \in \Omega_{B}}
= \; \cPlus(B)\,.
\]
This proves that axiom $(\alpha)$ holds for $\cPlus$, and the proof is complete.  
\end{proof}

\subsection{Lifting of axiom $(\beta)$} \label{SECT:lifting beta}

Here we prove that any partial choice satisfying axiom $(\beta)$ can be always lifted to a total choice still satisfying axiom $(\beta)$. 
To that end, we need the notion of the \emph{intersection graph} associated to a family of sets $\mathcal{S}$: this is the undirected graph whose nodes are the sets belonging to $\mathcal{S}$, and whose edges are the pairs of distinct intersecting sets $B,B' \in \mathcal{S}$ (i.e., such that $B \cap B' \neq \emptyset$).

\begin{theorem}\label{THM:lifting beta}
A partial choice has the $(\beta)$-lifting property if and only it satisfies axiom $(\beta)$.
\end{theorem}

\begin{proof}
Clearly, axiom $(\beta)$ holds for any choice that admits an extension to a total choice satisfying $(\beta)$.
Thus, it suffices to prove that any choice $c \colon \Omega \rightrightarrows U$ satisfying $(\beta)$ has the $(\beta)$-lifting property. 
For every $A \in \powPlus$, pick an element $u_{A} \in A$, subject only to the condition that $u_{A} \in c(A)$ whenever $A \in \Omega$. 
If $u_{A} \in \bigcup c[\Omega_{A}]$ (where $\Omega_{A} \defAs \{ B \in \Omega : B \subseteq A\}$), then let $\mathcal{C}_{A} \subseteq c[\Omega_{A}]$ be the connected component of the intersection graph associated to the family $c[\Omega_{A}]$ such that $u_{A} \in \bigcup \mathcal{C}_{A}$. 
Then, for $A \in \powPlus$, set
\[
\cPlus(A) \defAs \begin{cases}
\bigcup \mathcal{C}_{A} & \text{if } u_{A} \in \bigcup c[\Omega_{A}]\\
\{u_{A}\} & \text{otherwise}.
\end{cases}
\]
By definition, $\cPlus$ is a total contractive map on $\powPlus$ that is never empty-valued. 
In addition, if $A \in \Omega$, then $\cPlus(A) = \bigcup \mathcal{C}_{A} = c(A)$.
It follows that $\cPlus$ is a well-defined total choice that extends $c$. 

To complete the proof, we only need to show that $\cPlus$ satisfies $(\beta)$. 
Let $D,E \in \powPlus$ be such that $D \subseteq E$ and $\cPlus(D) \cap \cPlus(E) \neq \emptyset$. 
If $|\cPlus(D)| = 1$, then plainly $\cPlus(D) \subseteq \cPlus(E)$. 
On the other hand, if $|\cPlus(D)| > 1$, then $\cPlus(D) = \mathcal{C}_{D}$, hence $\cPlus(E) = \mathcal{C}_{E}$ and $\mathcal{C}_{D} \subseteq \mathcal{C}_{E}$. 
Thus, we obtain again $\cPlus(D) \subseteq \cPlus(E)$, as claimed. 
\end{proof}


\newcommand{\env}[2]{\mathsf{env}_{#1}(#2)}

\subsection{Lifting of \textsf{WARP}} \label{SECT:lifting of WARP}

Finally, we characterize choices that have the \textsf{WARP}-lifting property.
This characterization will be obtained in terms of the existence of a suitable Noetherian total preorder on the collection of the Euler's regions of the union of the choice domain with its image under the given choice. 
(Recall that a \emph{preorder} is a binary relation that is reflexive and transitive. Further, a relation $R$ on $X$ is \emph{Noetherian} if the converse relation $R^{-1}$ is well-founded, i.e., if every nonempty subset of $X$ has an $R$-maximal element.)

Thus, let $c \colon \Omega \rightrightarrows U$ be a partial choice. 
Denote by $\mathcal{E}$ the Euler's diagram of the family 
$
\Omega^{^{+}} \defAs \Omega \cup c[\Omega]\,,
$
namely the partition \\
\centerline{$
\mathcal{E} \defAs \left\{\bigcap \Gamma \:\setminus\: \bigcup (\Omega^{^{+}} \setminus \Gamma) : \emptyset \neq \Gamma \subseteq \Omega^{^{+}} \right\} \setminus \{\emptyset\}
$}
of $\bigcup \Omega^{^{+}}$ formed by all the nonempty sets of the form $\bigcap \Gamma \:\setminus\: \bigcup (\Omega^{^{+}} \setminus \Gamma)$, for $\emptyset \neq \Gamma \subseteq \Omega^{^{+}}$.
Further, for each $A \in \powPlus$, denote by $\env{\mathcal{E}}{A}$ the \emph{envelope of $A$ in $\mathcal{E}$}, namely, the collection of regions in $\mathcal{E}$ intersecting $A$; formally, 
$
\env{\mathcal{E}}{A} \defAs \{E \in \mathcal{E} : E \cap A \neq \emptyset\}\,.
$
Observe that, for each $B \in \Omega$, we have $\env{\mathcal{E}}{B} = \{E \in \mathcal{E} : E \subseteq B\}$.
It turns out that the choice $c$ can be lifted to a total choice satisfying \textsf{WARP} if and only if there exists a suitable Noetherian total preorder $\lesssim$ on $\mathcal{E}$ such that\\ 
\centerline{
$
E \subseteq B  \text{ and } E' \subseteq c(B) \text{ ~~~(for some  $B \in \Omega$)}  \qquad  \Longrightarrow \qquad E \lesssim E'  \,.
$}
More precisely, we have:

\begin{theorem} \label{THM:lifting WARP}
A partial choice $c \colon \Omega \rightrightarrows U$ has the \textsf{WARP}-lifting property if and only if there exists a total Noetherian preorder $\lesssim$ on the collection $\mathcal{E}$ of Euler's regions of $\Omega \cup c[\Omega]$ such that, for all $B \in \Omega$ and $E,E' \in \mathcal{E}$, the following conditions hold:
\begin{enumerate}[label=(\alph*)]
\item\label{aTheorem3} if $E \subseteq B$ and $E' \subseteq c(B)$, then $E \lesssim E'$;

\item\label{bTheorem3} if 
$E$ is $\lesssim$-maximal in $\env{\mathcal{E}}{B}$, then $E \subseteq c(B)$.
\end{enumerate}
\end{theorem}
 
The (constructive) proof of Theorem~\ref{THM:lifting WARP} is technical and rather long.
In order to maintain focus on the main goal of the paper, we omit including it here.
The interested reader can find it in the extended version of the paper available at the following URL:~~ \texttt{https://arxiv.org/abs/1708.06121}.


\section{The satisfiability problem in presence of a choice correspondence}\label{SECT:satProb}

We are now ready to define the syntax and semantics of the Boolean set-theoretic language extended with a choice correspondence, denoted by $\BSTC$, of which we shall study the satisfiability problem.

\subsection{Syntax of $\BSTC$}
The language $\BSTC$ involves
\begin{itemize}
\item denumerable collections $\mathcal{V}_{0}$ and $\mathcal{V}_{1}$ of \emph{individual} and \emph{set} variables, respectively;

\item the constant $\varnothing$ (empty set);

\item operation symbols: $\cdot\cup\cdot$, $\cdot\cap\cdot$, $\cdot\setminus\cdot$ , $\{\cdot\}$, $\choice(\cdot)$ (choice map);

\item predicate symbols: $\cdot=\cdot$, $\cdot\subseteq\cdot$, $\cdot\in\cdot$.
\end{itemize}
\emph{Set terms} of $\BSTC$ are recursively defined as follows:
\begin{itemize}
\item set variables and the constant $\varnothing$ are set terms;

\item if $T, T_{1},T_{2}$ are set terms and $x$ is an individidual variable, then
$
T_{1} \cup T_{2},~~ T_{1} \cap T_{2},~~ T_{1} \setminus T_{2},~~ \choice(T),~~ \{x\}
$
are set terms.
\end{itemize}
The \emph{atomic formulae} (or \emph{atoms}) of $\BSTC$ have one of the following two forms
$
T_{1} = T_{2}, ~~  T_{1} \subseteq T_{2}, 
$
where $T_{1},T_{2}$ are set terms. Atoms and their negations are called \emph{literals}.\\[.1cm]
Finally, \emph{$\BSTC$-formulae} are propositional combinations of $\BSTC$-atoms by means of the usual logical connectives $\land$, $\lor$, $\lnot$, $\implies$, $\iff$.

We regard $\{x_{1},\ldots,x_{k}\}$ as a shorthand for the set term $\{x_{1}\} \cup \ldots \cup \{x_{k}\}$. Likewise, $x \in T$ and $x = y$ are regarded as shorthands for $\{x\} \subseteq T$ and $\{x\} = \{y\}$, respectively.

\emph{Choice terms} are $\BSTC$-terms of type $\choice(T)$, whereas \emph{choice-free terms} are $\BSTC$-terms which do not involve the choice map $\choice$ (at any level of nesting). We refer to $\BSTC$-formulae containing only choice-free terms as $\BSTC^{-}$-formulae.\footnote{Up to minor syntactic differences, $\BSTC^{-}$-formulae are essentially \textsf{2LSS}-formulae, whose decision problem has been solved (see, for instance, \cite[Exercise 10.5]{CanOmoPol01}).}

\subsection{Semantics of $\BSTC$}

We first describe the \emph{unrestricted semantics of $\BSTC$}, when the choice operator is not required to satisfy any particular consistency axiom.

A \emph{set assignment} is a pair $\model=(U,M)$, where $U$ is any nonempty collection of objects, called the \emph{domain} or \emph{universe} of $\model$, and $M$ is an assignment over the variables of $\BSTC$ such that
\begin{itemize}
\item $\M x \in U$, for each individual variable $x \in \mathcal{V}_{0}$;

\item $\M \varnothing \defAs \emptyset$;

\item $\M X \subseteq U$, for each set variable $X \in \mathcal{V}_{1}$;

\item $\M \choice$ is a total choice correspondence over $U$.
\end{itemize}
Then, recursively, we put
\begin{itemize}
\item $\M{(T_{1} \otimes T_{2})} \defAs \M T_{1} \otimes \M T_{2}$, where $T_{1},T_{2}$ are set terms and $\otimes \in \{\cup,\cap,\setminus\}$;

\item $\M{\{x\}} \defAs  \{\M x\}$, where $x$ is an individual variable;

\item $\M {(\choice(T))} \defAs \M\choice(\M T)$, where $T$ is a set term.
\end{itemize}
Satisfiability of any $\BSTC$-formula $\psi$ by $\model$ (written $\model \models \psi$) is defined by putting
\[\begin{array}{rclrcr}
\model &\models& T_{1} \star T_{2}   &\qquad\text{iff}\qquad& \M T_{1} \star \M T_{2}\,,
%
\end{array}
\]
for $\BSTC$-atoms $T_{1} \star T_{2}$ (where $T_{1},T_{2}$ are set terms and $\star \in \{=,\subseteq\}$), and by interpreting logical connectives according to their classical meaning.

For a $\BSTC$-formula $\psi$,  if $\model \models \psi$ (i.e.,
$\model$ \emph{satisfies} $\psi$), then $\model$ is said to be a
\emph{$\BSTC$-model for $\psi$}.  A $\BSTC$-formula is said to be
\emph{satisfiable} if it has a $\BSTC$-model.  
Two $\BSTC$-formulae $\varphi$ and $\psi$ are \emph{equivalent} if they share exactly the same $\BSTC$-models; they are \emph{equisatisfiable} if one is satisfiable if and only if so is the other (possibly by different $\BSTC$-models). 

The \emph{satisfiability problem} (or \emph{decision problem}) for $\BSTC$ asks for an effective procedure (or \emph{decision procedure}) to establish whether any given $\BSTC$-formula is satisfiable or not.

\smallskip

We shall also address the satisfiability problem for $\BSTC$ under other semantics: specifically, the \emph{$(\alpha)$-semantics}, the \emph{$(\beta)$-semantics}, and the \textsf{WARP}-\emph{semantics} (whose satisfiability relations are denoted by $\models_{\alpha}$, $\models_{\beta}$, and $\models_{\textsf{WARP}}$, respectively). These differ from the unrestricted semantics in that the interpreted choice map $\M \choice$ is required to satisfy axiom $(\alpha)$ in the first case, axiom $(\beta)$ in the second case, and axioms $(\alpha)$ and $(\beta)$ conjunctively (namely \textsf{WARP}) in the latter case.


\smallskip

\subsection{The decision problem for $\BSTC$-formulae}\label{appendixDecProc}
%

The satisfiability problem for $\BSTC^{-}$- and $\BSTC$-formulae under the various semantics are \NP-hard, as the satisfiability problem for propositional logic can readily be reduced to any of them (in linear time).
In the cases of $(\alpha)$- and \textsf{WARP}-semantics, we shall prove \NP-completeness only under the additional hypothesis that the number of choice terms is constant, otherwise, in both cases, we have to content ourselves with a \NEXP\ complexity. As a by-product, it will follow that the satisfiability problem for $\BSTC^{-}$ is \NP-complete. On the other hand, we shall prove that the satisfiability problem for $\BSTC$-formulae under the unrestricted and the $(\beta)$-semantics can be reduced polynomially to the satisfiability problem for $\BSTC^{-}$-formulae, thereby proving their \NP-completeness.

\smallskip

Let $\varphi$ be a $\BSTC$-formula, $\mathsf{V}_{0} \subseteq \mathcal{V}_{0}$ and $\mathsf{V}_{1} \subseteq \mathcal{V}_{1}$ the collections of individual and set variables occurring in $\varphi$, respectively, and $\terms$ the collection of the set terms occurring in $\varphi$. For convenience, we shall assume that $\varnothing \in \terms$. Let also
\begin{equation}\label{choiceLiterals}
\choice(T_{1}),~\ldots,~\choice(T_{k})
\end{equation}
be the distinct choice terms occurring in $\varphi$, with $k \geqslant 0$ (when $k=0$, $\varphi$ is a $\BSTC^{-}$-formula). 
 
Without loss of generality, we may assume that $\varphi$ is in \emph{choice-flat form}, namely that all the terms $T_{1},\ldots,T_{k}$ in (\ref{choiceLiterals}) are choice-free. In fact, if this were not the case, then, for each choice term $\choice(T)$ in $\varphi$ occurring inside the scope of a choice symbol and such that $T$ is choice-free, we could replace in $\varphi$ all occurrences of $\choice(T)$ by a newly introduced variable $X_{T}$ and add the conjunct $X_{T} = \choice(T)$ to $\varphi$, until no choice term is left which properly contains a choice subterm. It is an easy matter to check that the resulting formula is in choice-flat form, it is equisatisfiable with $\varphi$ (under any of our semantics), and its size is linear in the size of $\varphi$. 

Without disrupting satisfiability, we may add to $\varphi$ the following formulae:
\begin{description}
\item[\emph{choice conditions}:] $\varnothing \neq \choice(T_{i}) \;\land\; \choice(T_{i}) \subseteq T_{i}$, ~~for $i=1,\ldots,k$;

\item[\emph{single-valuedness conditions}:] $T_{i} = T_{j} \implies \choice(T_{i}) = \choice(T_{j})$, ~~for all distinct $i,j = 1,\ldots,k$,
\end{description}
since they are plainly true in any $\BSTC$-assignment.
In this case the size of $\varphi$ could have up to a quadratic increase. However, the total number of terms remains unchanged.\footnote{See Footnote~\ref{FootnoteFiner}.} For the sake of simplicity, we shall assume that $\varphi$ includes its choice and single-valuedness conditions, and thereby say that it is \emph{complete}. 

Notice that the above considerations hold irrespectively of the semantics adopted. 

In the sections which follow, we study the satisfiability problem for complete $\BSTC$-formulae under the various semantics described earlier. We start our course with the \textsf{WARP}-semantics.

\subsubsection{\textsf{WARP}-semantics}\label{WARPsemantics}

We first derive some necessary conditions for $\varphi$ to be satisfiable and later prove their sufficiency. Hence, to begin with, let us assume that $\varphi$ is satisfiable under the \textsf{WARP}-semantics and let $\model = (U,M)$ be a model for it. 
Let $\M{\mathcal{R}_{\varphi}}$ be the Euler's diagram of $\M{\terms} \defAs \{\M T : T \in \terms\}$.
Notice that, for each region $\rho \in \M{\mathcal{R}_{\varphi}}$ and term $T \in \terms$, either $\rho \subseteq \M T$ or $\rho \cap \M T = \emptyset$. Thus, to each $\rho \in \M{\mathcal{R}_{\varphi}}$,  there corresponds a Boolean map $\pi_{\rho} \colon \terms \rightarrow \{\true,\false\}$ over $\terms$ (where we have identified the truth values \textsf{true} and \textsf{false} with $\true$ and $\false$, respectively) such that
\begin{equation}\label{defPlaces}
\pi_{\rho}(T) = \true \iff \rho \subseteq \M T \iff \rho \cap \M T \neq \emptyset\,,~~~~\text{for } T \in \terms\,.
\end{equation}
Let $\M{\Pi_{\varphi}} \defAs \{\pi_{\rho} : \rho \in \M{\mathcal{R}_{\varphi}}\}$. Hence, we have:
\begin{enumerate}[label=\text{(\roman*)}]
\item\label{zeroPlace} $\pi(\varnothing) = \false$, for each $\pi \in \M{\Pi_{\varphi}}$;

\item\label{firstPlace} $\pi(T_{1} \cup T_{2}) = \pi(T_{1}) \lor \pi(T_{2})$, for each map $\pi \in \M{\Pi_{\varphi}}$ and set term $T_{1} \cup T_{2}$ in $\varphi$;

\item\label{secondPlace} $\pi(T_{1} \cap T_{2}) = \pi(T_{1}) \land \pi(T_{2})$, for each map $\pi \in \M{\Pi_{\varphi}}$ and set term $T_{1} \cap T_{2}$ in $\varphi$;

\item\label{thirdPlace} $\pi(T_{1} \setminus T_{2}) = \pi(T_{1}) \land \lnot \pi(T_{2})$, for each map $\pi \in \M{\Pi_{\varphi}}$ and set term $T_{1} \setminus T_{2}$ in $\varphi$.
\end{enumerate}
In addition, we have $\M T = \bigcup \{\rho \in \M{\mathcal{R}_{\varphi}} : \pi_{\rho}(T) = \true\}$, for every $T \in \terms$.

By uniformly replacing the atomic formulae in $\varphi$ with propositional variables, in such a way that different occurrences of the same atomic formula are replaced by the same propositional variable and different atomic formulae are replaced by distinct propositional variables, we can associate to $\varphi$ its \emph{propositional skeleton} $\mathsf{P}_{\varphi}$ (up to variables renaming). For instance, the propositional skeleton of
\begin{equation}\label{prop-example}
((X = Y \setminus X \:\land\: Y = X \cup \choice(X_{1})) \Longrightarrow Z \neq \varnothing) \Longleftrightarrow (X = Y \setminus X \Longrightarrow (Y = X \cup \choice(X_{1})) \Longrightarrow Z = \varnothing))
\end{equation}
is the propositional formula 
\[
((P_{1} \land P_{2}) \Longrightarrow \neg P_{3}) \Longleftrightarrow (P_{1} \Longrightarrow (P_{2} \Longrightarrow P_{3}))\,.
\]
Plainly, a necessary condition for $\varphi$ to be satisfiable (by a $\BSTC$-model) is that its skeleton $\mathsf{P}_{\varphi}$ is propositionally satisfiable (however, the converse does not hold in general). 

\smallskip
A collection $\mathcal{A}$ of atoms of $\varphi$ is said to be \emph{promising} for $\varphi$ if the valuation which maps to \textsf{true} the propositional variables corresponding to the atoms in $\mathcal{A}$ and to \textsf{false} the remaining ones satisfies the propositional skeleton $\mathsf{P}_{\varphi}$. For instance, in the case of (\ref{prop-example}), all collections of its  atoms not containing both $X = Y \setminus X$ and $Y = X \cup \choice(X_{1})$ are promising for (\ref{prop-example}).

Let $\mathcal{A}_{\varphi}^{+}$ be the collection of the atoms in $\varphi$ satisfied by $\model$ and $\mathcal{A}_{\varphi}^{-}$ the collection of the remaining atoms in $\varphi$, namely those that are disproved by $\model$. It can be easily checked that $\mathcal{A}_{\varphi}^{+}$ is promising. In addition, for every atom $T_{1} = T_{2}$ in $\varphi$ and $\pi \in \M{\Pi_{\varphi}}$, we have $\pi(T_{1}) = \pi(T_{2})$ if and only if $T_{1} = T_{2}$ is in $\mathcal{A}_{\varphi}^{+}$. Likewise, for every atom $T_{1} \subseteq T_{2}$ in $\varphi$ and $\pi \in \M{\Pi_{\varphi}}$, we have $\pi_{\rho}(T_{1}) \leqslant \pi_{\rho}(T_{2})$ if and only if $T_{1} \subseteq T_{2}$ is in $\mathcal{A}_{\varphi}^{+}$. Thus, in particular, for every atom $T_{1} = T_{2}$ in $\mathcal{A}_{\varphi}^{-}$, there exists a map $\pi \in \M{\Pi_{\varphi}}$ such that $\pi(T_{1}) \neq \pi(T_{2})$. Likewise, for every atom $T_{1} \subseteq T_{2}$ in $\mathcal{A}_{\varphi}^{-}$, there exists a map $\pi \in \M{\Pi_{\varphi}}$ such that $\pi(T_{1}) > \pi(T_{2})$.

%
%
\begin{definition}[Places]\label{places}\rm
Any Boolean map $\pi$ on $\terms$\, for which the above properties \ref{zeroPlace}--\ref{thirdPlace} hold for each set term in $\varphi$ is called a \emph{place for $\varphi$}.

For a given set $\mathcal{A}$ of atoms occurring in $\varphi$, a set $\Pi$ of places for $\varphi$ is $\mathcal{A}$-\emph{ample} if 
\begin{enumerate}[label=\text{(\roman*)},start=5]
\item\label{fourthPlace} $\pi(T_{1}) = \pi(T_{2})$ for each $\pi \in \Pi$, provided that the atom $T_{1} = T_{2}$ is not in $\mathcal{A}$;

\item\label{fifthPlace} $\pi(T_{1}) \leqslant \pi(T_{2})$ for each $\pi \in \Pi$, provided that the atom $T_{1} \subseteq T_{2}$ is not in $\mathcal{A}$;

\item\label{sixthPlace} $\pi(T_{1}) \neq \pi(T_{2})$ for some $\pi \in \Pi$, if the atom $T_{1} = T_{2}$ belongs to $\mathcal{A}$;

\item\label{seventhPlace} $\pi(T_{1}) > \pi(T_{2})$ for some $\pi \in \Pi$, if the atom $T_{1} \subseteq T_{2}$ belongs to $\mathcal{A}$.
\eod
%
\end{enumerate}
\end{definition}
Notice that conditions \ref{zeroPlace}--\ref{fifthPlace} are universal, whereas conditions \ref{sixthPlace}--\ref{seventhPlace} are existential.

The considerations made just before Definition~\ref{places} yield that $\M{\Pi_{\varphi}}$ is an $\mathcal{A}_{\varphi}^{-}$-ample set of places for $\varphi$. However, in order to later establish some tight complexity results, it is convenient to enforce a polynomial bound for the cardinality of the set of places in terms of the size $|\varphi|$ of $\varphi$ (where, for instance, $|\varphi|$ could be defined as the number of nodes in the syntax tree of $\varphi$). We do this as follows: for each atom $T_{1} = T_{2}$ (resp., $T_{1} \subseteq T_{2}$) in $\mathcal{A}_{\varphi}^{-}$, we select a place $\pi_{\rho}$, with $\rho \in \M{\mathcal{R}_{\varphi}}$ such that $\rho \subseteq \M{T_{1}} \Longleftrightarrow \rho \nsubseteq \M{T_{2}}$ (resp., $\rho \subseteq \M{T_{1}} \setminus \M{T_{2}}$) holds, and call their collection $\Pi_{1}$. Plainly, we have $|\Pi_{1}| \leq |\mathcal{A}_{\varphi}^{-}|$. Notice that $\Pi_{1}$ is $\mathcal{A}_{\varphi}^{-}$-ample.\footnote{\label{FootnoteFiner}A finer construction would yield an $\mathcal{A}_{\varphi}^{-}$-ample set of places $\Pi_{1}'$ such that $|\Pi_{1}'| \leq |\terms|$. Notice that $|\Pi_{1}| = \mathcal{O}(|\Pi_{1}'|^{2})$.}

Conditions \ref{zeroPlace}--\ref{thirdPlace} take care of the structure of set terms in $\varphi$ but those of the form $\{x\}$ or $\choice(T)$, conditions \ref{fourthPlace} and \ref{fifthPlace} take care of the atoms in $\varphi$ deemed to be positive, whereas conditions \ref{sixthPlace} and \ref{seventhPlace} take care of the remaining atoms in $\varphi$, namely those deemed to be negative.

To take care of set terms of the form $\{x\}$ in $\varphi$, we observe that, for every $x \in \mathsf{V}_{0}$, there exists a unique Euler's region $\rho_{x} \in \M{\mathcal{R}_{\varphi}}$ such that $\M x \in \rho_{x}$. Let  $\pi^{x}$ be the place corresponding to $\rho_{x}$ according to (\ref{defPlaces}), namely $\pi^{x} \defAs \pi_{\rho_{x}}$, and put $\Pi_{2} \defAs \{\pi^{x} : x \in \mathsf{V}_{0}\}$. 

\begin{definition}[Places at variables]\label{placeAtVariables}\rm
Let $x$ be an individual variable occurring in $\varphi$. A \emph{place (for $\varphi$) at the variable $x$} is any place $\pi$ for $\varphi$ such that $\pi(\{x\}) = \true$.\eod
\end{definition}

Next, we take care of choice terms. Thus, let $\Omega \defAs \{\M{T_{1}},\ldots,\M{T_{k}}\}$ and let ${\mathcal{E}}$ be the Euler's diagram of $\Omega \cup \M\choice[\Omega]$. Notice that each region in ${\mathcal{E}}$ is a disjoint union of regions in $\M{\mathcal{R}_{\varphi}}$; moreover, the partial choice $\M\choice\restrict{\Omega}$ over the choice domain $\Omega$ enjoys the $\textsf{WARP}$-lifting property. Thus, by Theorem~\ref{THM:lifting WARP}, there exists a total Noetherian preorder $\lesssim$ on ${\mathcal{E}}$ such that, for all $\sigma',\sigma'' \in {\mathcal{E}}$ and $\M{T} \in \Omega$, we have:
\begin{enumerate}[label=(\Alph*)]
\item\label{Aabove}  if $\sigma' \subseteq \M{T}$ and $\sigma'' \subseteq \M\choice(\M{T})$, then $\sigma' \lesssim \sigma''$;

\item\label{Babove}  if $\sigma'$ is $\lesssim$-maximal in $\mathsf{env}_{{\mathcal{E}}}(\M{T})$, then $\sigma' \subseteq \M\choice(\M{T})$.
\end{enumerate}
For each region $\sigma \in \mathcal{E}$, let us select a place $\pi_{\rho}$, such that $\rho \in \M{\mathcal{R}_{\varphi}}$ and $\rho \subseteq \sigma$, and call their collection $\Pi_{3}$. Set $\Pi \defAs \Pi_{1} \cup \Pi_{2} \cup \Pi_{3}$. Plainly, $|\Pi| \leqslant |\mathcal{A}| + |\mathsf{V}_{0}| + 2^{k}$, $\Pi$ is $\mathcal{A}_{\varphi}^{-}$-ample, and $\pi^{x}$ is the sole place in $\Pi$ at the variable $x$, for each individual variable $x$ in $\varphi$. To ease notation, for $\pi \in \Pi$, $\Pi' \subseteq \Pi$, and $T \in \terms$,  we shall also write (i) $\pi \subseteq T$ for $\pi(T) = \true$, (ii) $\Pi' \subseteq T$ for $\pi' \subseteq T$, for every $\pi' \in \Pi'$, (iii) $\Pi' \ni\in T$ for $\pi' \subseteq T$, for some $\pi' \in \Pi'$.
\renewcommand{\P}{\mathcal{P}}
For each $\sigma \in \mathcal{E}$, let $\Pi_{\sigma} \defAs \{\pi_{\rho} : \rho \in \M{\mathcal{R}_{\varphi}} \text{ and } \rho \subseteq \sigma\}$, and call $\P$ their collection. 
Then, by \ref{Aabove} and \ref{Babove} above, there exists a total Noetherian preorder $\precsim$ on $\P$ such that, for all $\Pi',\Pi'' \in \P$ and $i \in \{1,\ldots,k\}$, the following conditions hold:
\begin{enumerate}[label=(\Alph*')]
\item\label{AabovePrime}  if $\Pi' \subseteq T_{i}$ and $\Pi'' \subseteq \choice(T_{i})$, then $\Pi' \precsim \Pi''$;


\item\label{BabovePrime}  if $\Pi'$ is $\precsim$-maximal in $\{\Pi^{*} \in \P : 
\Pi^{*} \ni\in T_{i}\}$, then $\Pi' \subseteq \choice(T_{i})$.
%
\end{enumerate}



Summing up, we have the following result:
\begin{lemma}\label{necessity2LSS}
Let $\varphi$ be a $\BSTC$-formula in choice-flat form, $\mathsf{V}_{0}$ the set of individual variables occurring in it, and $\choice(T_{1}),~\ldots,~\choice(T_{k})$ the choice terms occurring in it. If $\varphi$ is satisfiable under the \textsf{WARP}-semantics, then there exist an $\mathcal{A}$-ample set $\Pi$ of places for $\varphi$ such that $|\Pi| \leq |\mathcal{A}| + |\mathsf{V}_{0}| + 2^{k}$, for some promising set $\mathcal{A}$ of atoms in $\varphi$, and a map $x \mapsto \pi^{x}$ from $\mathsf{V}_{0}$ into $\Pi$ such that $\pi^{x}$ is the sole place in $\Pi$ at the variable $x$, for $x \in \mathsf{V}_{0}$. In addition, if
$
\Pi_{\choice} \defAs \big\{\pi \in \Pi : \pi \subseteq T_{i} \text{ or } \pi \subseteq \choice(T_{i})\text{, for some } i \in \{1,\ldots,k\} \big\}\,,
$
$\sim_{\choice}$ is the equivalence relation on $\Pi_{\choice}$ such that
\\
\centerline{$
\pi \sim_{\choice} \pi' \quad \Longleftrightarrow \quad \pi(T_{i}) = \pi'(T_{i}) \text{ ~and~ } \pi(\choice(T_{i})) = \pi'(\choice(T_{i}))\text{\,,~ for every } i = 1,\ldots,k\,,
$}
and $\P \defAs \Pi_{\choice}/\sim_{\choice}$, then there exists a total Noetherian preorder $\precsim$ on $\P$ such that conditions \ref{AabovePrime} and \ref{BabovePrime} are satisfied for all $\Pi',\Pi'' \in \P$ and $i \in \{1,\ldots,k\}$.
\eod
\end{lemma}

Next we show that the conditions in the preceding lemma are also sufficient for the satisfiability of our $\BSTC$-formula $\varphi$ under the \textsf{WARP}-semantics. Thus, let $\Pi$, $\mathcal{A}$, $x \mapsto \pi^{x}$, $\Pi_{\choice}$, $\sim_{\choice}$, and $\P$ be such that the  conditions in Lemma~\ref{necessity2LSS} are satisfied. Let $U_{_{\Pi}}$ be any set of cardinality $|\Pi|$, and $\pi \mapsto a_{\pi}$ any injective map from $\Pi$ onto $U_{_{\Pi}}$. We define an interpretation $M_{_{\Pi}}$ over the variables in $\mathsf{V}_{0} \cup \mathsf{V}_{1}$ and the choice terms $\choice(T_{1}),~\ldots,~\choice(T_{k})$ occurring in $\varphi$, as if they were set variables, by putting:
\newcommand{\piM}[1]{#1^{\scriptscriptstyle M_{\Pi}}}
\begin{equation}\label{finalModel}
\begin{array}{rcll}
\piM x &\defAs& a_{\pi^{x}}&\text{for } x \in \mathsf{V}_{0}\\
\piM X &\defAs& \{a_{\pi} \in U_{_{\Pi}} : \pi \in \Pi \land \pi \subseteq X \}~~~~~&\text{for } X \in \mathsf{V}_{1}\\
\piM{(\choice(T_{i}))} &\defAs& \{a_{\pi} \in U_{_{\Pi}} : \pi \in \Pi \land \pi \subseteq \choice(T_{i}) \}~~~~~&\text{for } i =1,\ldots,k\,.
\end{array}
\end{equation}
Notice that the choice map $\choice$ is not interpreted by $M_{_{\Pi}}$. However, it is not hard to check that the set assignment $\model_{_{\Pi}} \defAs (U_{_{\Pi}},M_{_{\Pi}})$ satisfies exactly the atoms in $\mathcal{A}$, provided that choice terms in $\varphi$ are regarded as set variables with no internal structure, rather than as compound terms. Indeed, after putting $\piM \pi \defAs \{a_{\pi}\}$, for $\pi \in \Pi$, we have $\piM T = \bigcup_{\pi \subseteq T} \piM \pi$, for every $T \in \terms$. Thus, if $T_{1} = T_{2}$ does not belong to $\mathcal{A}$, then by condition \ref{fourthPlace} we have $\piM{T_{1}} = \piM{T_{2}}$, whereas if $T_{1} = T_{2}$  belongs to $\mathcal{A}$, then by  \ref{sixthPlace} we have $\piM{T_{1}} \neq \piM{T_{2}}$. Analogously for atoms $T_{1} \subseteq T_{2}$ in $\varphi$.
Hence, by the promisingness of $\mathcal{A}$, $\model_{_{\Pi}}$ satisfies also $\varphi$.

Next we define $M_{_{\Pi}}$ also over the choice map $\choice$ in such a way that $\piM \choice$ is a total choice over $U$ that satisfies \textsf{WARP}, and $\model_{_{\Pi}} \models \varphi$ holds.
Let $\Omega_{_{\Pi}} \defAs \{\piM{T_{1}},\ldots,\piM{T_{k}}\}$. We begin by defining $\piM{\choice}$ over $\Omega_{_{\Pi}}$ in the most natural way, namely by putting $\piM{\choice}(\piM{T_{i}}) \defAs \piM{(\choice(T_{i}))}$, for $i = 1,\ldots,k$. By the choice and the single-valuedness conditions in $\varphi$, $\piM{\choice}$ so defined is a choice over the domain $\Omega_{_{\Pi}}$. In addition, it can be checked that the existence of a total Noetherian preorder on $\P$ satisfying conditions \ref{AabovePrime} and \ref{BabovePrime} yield the existence of a total Noetherian preorder on the collection of the Euler's regions of $\Omega_{_{\Pi}} \cup \piM{\choice}[\Omega_{_{\Pi}}]$ satisfying conditions \ref{aTheorem3} and \ref{bTheorem3} of Theorem~\ref{THM:lifting WARP}. Thus, the latter theorem readily implies that $\piM{\choice}$ has the \textsf{WARP}-lifting property, so that it can be extended to a total choice on $U$ satisfying \textsf{WARP}, proving that $\varphi$ is satisfiable under the \textsf{WARP}-semantics.

Together with Lemma~\ref{necessity2LSS}, the above argument yields the following result:
\begin{theorem}\label{WARPTheorem}
Let $\varphi$ be a complete $\BSTC$-formula in choice-flat form, $\mathsf{V}_{0}$ the set of individual variables occurring in it, and $\choice(T_{1}),~\ldots,~\choice(T_{k})$ the choice terms occurring in it. Then $\varphi$ is satisfiable under the \textsf{WARP}-semantics if and only if there exist 
(i) an $\mathcal{A}$-ample set $\Pi$ of places for $\varphi$ such that $|\Pi| \leq |\mathcal{A}| + |\mathsf{V}_{0}| + 2^{k}$, for some promising set $\mathcal{A}$ of atoms in $\varphi$, (ii) a map $x \mapsto \pi^{x}$ from $\mathsf{V}_{0}$ into $\Pi$, such that $\pi^{x}$ is the sole place in $\Pi$ at the variable $x$, for $x \in \mathsf{V}_{0}$, and (iii) a total Noetherian preorder $\precsim$ on $\P$, 
where $\P \defAs \Pi_{\choice}/\sim_{\choice}$ (with $\Pi_{\choice} \defAs \big\{\pi \in \Pi : \pi \subseteq T_{i} \text{ or } \pi \subseteq \choice(T_{i})\text{, for some } i \in \{1,\ldots,k\} \big\}\,,$ and $\sim_{\choice}$  the equivalence relation on $\Pi_{\choice}$ such that
$\pi \sim_{\choice} \pi'  \Longleftrightarrow  \pi(T_{i}) = \pi'(T_{i}) \text{ ~and~ } \pi(\choice(T_{i})) = \pi'(\choice(T_{i}))\text{\,,~ for } i = 1,\ldots,k$) such that the above conditions \ref{AabovePrime} and \ref{BabovePrime} are satisfied for all $\Pi',\Pi'' \in \P$ and $i \in \{1,\ldots,k\}$.
\eod
\end{theorem}

We have already observed that the satisfiability problem for $\BSTC$-formulae under the \textsf{WARP}-semantics is  \NP-hard. In addition, the previous theorem implies at once that it belongs to \NEXP. However, if we restrict to $\BSTC$-formulae with a constant number of choice terms, the disequality $|\Pi| \leq |\mathcal{A}| + |\mathsf{V}_{0}| + 2^{k}$ in Theorem~\ref{WARPTheorem} yields $|\Pi| = \mathcal{O}(|\mathcal{A}| + |\mathsf{V}_{0}|) =\mathcal{O}(|\varphi|)$, thereby providing the following complexity result:
\begin{theorem}
Under the \textsf{WARP}-semantics, the satisfiability problem for $\BSTC$-formulae with $\mathcal{O}(1)$ distinct choice terms is \NP-complete.\eod
\end{theorem}
\medskip
As a by-product, Theorem~\ref{WARPTheorem} yields a solution to the satisfiability problem for $BSTC^{-}$-formulae. Indeed, in the case of $BSTC^{-}$-formulae, Theorem~\ref{WARPTheorem} becomes:

\begin{theorem}\label{BSTC-Theorem}
Let $\varphi$ be a $\BSTC^{-}$-formula, and $\mathsf{V}_{0}$ the set of individual variables occurring in it. Then $\varphi$ is satisfiable if and only if there exist (i) an $\mathcal{A}$-ample set $\Pi$ of places for $\varphi$ such that $|\Pi| \leq |\mathcal{A}| + |\mathsf{V}_{0}|$, for some promising set $\mathcal{A}$ of atoms in $\varphi$, and (ii) a map $x \mapsto \pi^{x}$ from $\mathsf{V}_{0}$ into $\Pi$, such that $\pi^{x}$ is the sole place in $\Pi$ at the variable $x$, for $x \in \mathsf{V}_{0}$.
\eod
\end{theorem}

Hence, we also have:
\begin{theorem}\label{NPcompl2LSS}
The satisfiability problem for $\BSTC^{-}$-formulae is \NP-complete.\eod
\end{theorem}

\subsubsection{Unrestricted semantics}

%
%
As above, let $\varphi$ be a complete $\BSTC$-formula in choice flat-form.
Plainly, if $\varphi$ is satisfiable under unrestricted semantics, so is its \emph{$\BSTC^{-}$-reduction} obtained from $\varphi$ by regarding the choice terms as set variables with no internal structure.

Conversely, let us assume that the $\BSTC^{-}$-reduction $\varphi_{1}$ of $\varphi$ is satisfiable, and let $\model_{1} = (U,M_{1})$ be a model for $\varphi_{1}$. 
We define a total choice correspondence $c \colon \powPlus \rightrightarrows U$ on $U$ by putting, for every $A \in \powPlus$,
\begin{equation}\label{defMapc}
c(A) \defAs \begin{cases}
A & \text{if } A \notin \{ \M T_{1},\ldots, \M T_{k}\}\\
\M{(\choice(T_{i}))} & \text{if } A = \M T_{i}, \text{ for some } i = 1,\ldots,k.
\end{cases}
\end{equation}
Observe that, by the single-valuedness conditions present in $\varphi_{1}$, if $\M T_{i} = \M T_{j}$, for distinct $i$ and $j$, then $\M{(\choice(T_{i}))} = \M{(\choice(T_{j}))}$. Hence, the map $c$ is well-defined. In addition, the choice conditions yield that $c$ is indeed a total choice on $U$.

Let $\model = (U,M)$ be the set assignment differing from $\model_{1}$ only on the interpretation of the choice map symbol $\choice$, for which we have $\M \choice \defAs c$. Plainly, $\model$ coincides with $\model_{1}$ on all the choice-free terms in $\varphi$. In addition, since $c(\M T_{i}) = \M{(\choice(T_{i}))}$ for $i=1,\ldots,k$, the assignment $\model$ coincides with $\model_{1}$ on the remaining terms in $\varphi$ as well. Thus $\model \models \varphi$, proving that $\varphi$ is satisfiable when it admits a satisfiable $\BSTC^{-}$-reduction.

We have thus proved:
\begin{lemma}
Under unrestricted semantics, a complete $\BSTC$-formula in choice-flat form is satisfiable if and only if it admits a satisfiable $\BSTC^{-}$-reduction.\eod
\end{lemma}

In view of Theorem~\ref{NPcompl2LSS}, we can conclude: 
\begin{theorem}
Under unrestricted semantics, the satisfiability problem for $\BSTC$-formulae is \NP-complete.\eod
\end{theorem}


\subsubsection{$(\beta)$-semantics}
Let us now assume that our complete $\BSTC$-formula in choice-flat form $\varphi$ is satisfiable under the $(\beta)$-semantics, and let $\model=(U,M)$ be a model for it, where now $\M \choice$ is a choice satisfying the axiom $(\beta)$. Then the model $\model$ satisfies all the following instances of the axiom $(\beta)$:
\begin{description}
\item[$(\beta)$\emph{-conditions}:] $\big(T_{i} \subseteq T_{j} \: \wedge \: \choice(T_{i}) \cap \choice(T_{j}) \neq \emptyset \big) \;\; \Longrightarrow \;\; \choice(T_{i}) \subseteq \choice(T_{j})$, ~~for $i,j = 1,\ldots,k$.
\end{description}

Let $\varphi_{\beta}$ be the $\BSTC^{-}$-formula obtained by adding the $(\beta)$-conditions to the $\BSTC^{-}$-$\beta$-reduction of $\varphi$, while regarding the choice terms in it just as set variables (with no internal structure).  We call the formula $\varphi_{\beta}$ the $\BSTC^{-}$-\emph{$\beta$-reduction of $\varphi$}. Notice that $|\varphi_{\beta}| = \mathcal{O}(|\varphi|^{2})$. In addition, $\varphi_{\beta}$ is plainly satisfiable.


Conversely, let $\varphi_{\beta}$ be satisfiable and let $\model_{\beta} = (U,M_{\beta})$ be a model for $\varphi_{\beta}$. Let $\Omega_{\beta} \defAs \{\Mbeta{T_{i}} : i = 1, \ldots,k\}$ and $c_{\beta}$ be a map over $\Omega_{\beta}$ such that $c_{\beta}(\Mbeta{T_{i}}) =  \Mbeta{(\choice(T_{i}))}$, for $i = 1, \ldots,k$. From the choice and single-valuedness conditions, $c_{\beta}$ is a choice over the domain $\Omega_{\beta}$. In addition, by the $\beta$-conditions present in $\varphi_{\beta}$, the choice $c_{\beta}$ satisfies the $(\beta)$-axiom. Hence, Theorem~\ref{THM:lifting beta} yields that $c_{\beta}$ has the $(\beta)$-lifting property, i.e., there is a total choice $\cPlus_{\beta} \colon \powPlus \rightrightarrows U$ extending $c_{\beta}$ and satisfying the $(\beta)$-axiom. Let $\model = (U,M)$ be the set assignment differing from $\model_{\beta}$ only on the interpretation of the choice symbol $\choice$, for which we have $\M \choice = \cPlus_{\beta}$. It is not hard to check that $\model \models_{\beta} \varphi$.

Thus, we have:
\begin{lemma}
A complete $\BSTC$-formula in choice-flat form is satisfiable under the $(\beta)$-semantics if and only if it admits a satisfiable $\BSTC^{-}$-$\beta$-reduction.\eod
\end{lemma}

Since the size of the $\BSTC^{-}$-$\beta$-reduction of a given $\BSTC$-formula $\psi$ is at most quadratic in the size of $\psi$, in view of Theorem~\ref{NPcompl2LSS} we can conclude: 
\begin{theorem}
Under the $(\beta)$-semantics, the satisfiability problem for $\BSTC$-formulae is \NP-complete.\eod
\end{theorem}

\subsubsection{$(\alpha)$-semantics}
Finally, we assume that our complete $\BSTC$-formula in choice-flat form $\varphi$ is satisfiable under the $(\alpha)$-semantics. Let $\model=(U,M)$ be a model for it, where now $\M \choice$ is a choice satisfying the axiom $(\alpha)$. Then, the model $\model$ satisfies also the following instances of the axiom $(\alpha)$:
\begin{description}
\item[$(\alpha)$\emph{-conditions}:] $T_{i} \subseteq T_{j}  \;\; \Longrightarrow \;\; T_{i} \cap \choice(T_{j}) \subseteq \choice(T_{j})$, ~~for $i,j = 1,\ldots,k$.
\end{description}
In addition, let $\Omega \defAs \{\M{T_{i}} : i = 1, \ldots, k\}$. Then, plainly, the partial choice correspondence $\M\choice\restrict{\Omega}$ over the choice domain $\Omega$ has the $(\alpha)$-lifting property. Hence, by Theorem~\ref{THM:lifting alpha}\ref{c}, for every $\emptyset \neq \mathcal{B} \subseteq \Omega$ such that $\mathcal{B}$ is $\subseteq$-closed w.r.t.\ $\Omega$, we have
\begin{equation}\label{nonEmptynessSem}
\bigcup \mathcal{B} \setminus \bigcup\nolimits_{B \in \mathcal{B}} \overline{c}(B) \neq \emptyset\,.
\end{equation}
In fact, (\ref{nonEmptynessSem}) holds for every $\emptyset \neq \mathcal{B} \subseteq \Omega$, irrespectively of whether $\mathcal{B}$ is $\subseteq$-closed w.r.t.\ $\Omega$ or not, as can be easily checked. Thus, $\model$ satisfies also the following further conditions:
\begin{description}
\item[\emph{nonemptiness conditions}:] $\bigcup_{i \in I} T_{i} \setminus \bigcup_{i \in I}(T_{i} \setminus \choice(T_{i}))  \neq \varnothing$, ~~for each $\emptyset \neq I \subseteq \{1,\ldots,k\}$.
\end{description}

Let $\varphi_{\alpha}$ be the $\BSTC^{-}$-formula obtained by adding the $(\alpha)$- and the nonemptiness conditions to the $\BSTC^{-}$-reduction of $\varphi$, while regarding the choice terms in it just as set variables (with no internal structure).  We call the formula $\varphi_{\alpha}$ the $\BSTC^{-}$-\emph{$\alpha$-reduction of $\varphi$}. Notice that $|\varphi_{\alpha}| = \mathcal(|\varphi|^{2} + k\cdot 2^{k})$. In addition, $\varphi_{\alpha}$ is plainly satisfiable.

%
%

Conversely, let us assume that $\varphi_{\alpha}$ is satisfiable and let $\model_{\alpha} = (U,M_{\alpha})$ be a model for it. Let $\Omega_{\alpha} \defAs \{\Malpha{T_{i}}  : i = 1, \ldots,k\}$ and $c_{\alpha}$ be a map over $\Omega_{\alpha}$ such that $c_{\alpha}(\Malpha{T_{i}}) =  \Malpha{X_{i}}$, for $i = 1, \ldots,k$. As before, thanks to the choice and the single-valuedness conditions, $c_{\alpha}$ is a choice over the domain $\Omega_{\alpha}$. In addition, from the $(\alpha)$- and the nonemptiness conditions, Theorem~\ref{THM:lifting alpha} yields that $c_{\alpha}$ has the $(\alpha)$-lifting property, i.e., there is a total choice $\cPlus_{\alpha} \colon \powPlus \rightrightarrows U$ extending $c_{\alpha}$ and satisfying the $(\alpha)$-axiom. Let $\model = (U,M)$ be the set assignment differing from $\model_{\alpha}$ only on the interpretation of the choice symbol $\choice$, for which we have $\M \choice = \cPlus_{\alpha}$. It is routine to check that $\model \models_{\alpha} \varphi$.

Hence, we have:
\begin{lemma}
Under the $(\alpha)$-semantics, a complete $\BSTC$-formula in choice-flat form is satisfiable if and only if it admits a satisfiable $\BSTC^{-}$-$\alpha$-reduction.\eod
\end{lemma}

As observed earlier, the size of $\varphi_{\alpha}$ is $\mathcal{O}(|\varphi|^{2} + k\cdot 2^{k})$. Thus, in general, the satisfiability problem for $\BSTC$-formulae under the $(\alpha)$-semantics is in \NEXP. However, if the number of distinct choice terms is restricted to be $\mathcal{O}(1)$, we have $|\varphi_{\alpha}| = \mathcal{O}(|\varphi|^{2})$ and therefore, in view of Theorem~\ref{NPcompl2LSS}, we have:
\begin{theorem}
Under the $(\alpha)$-semantics, the satisfiability problem for $\BSTC$-formulae with $\mathcal{O}(1)$ distinct choice terms is \NP-complete.\eod
\end{theorem}

\section{Conclusions} \label{SECT:Conclusions}

In this paper we have initiated the study of the satisfiability problem for unquantified formulae of an elementary fragment of set theory enriched with a choice correspondence symbol. 
Apart from the obvious theoretical reasons that motivate this approach, our analysis has its roots in applications within the field of social and individual choice theory.
In fact, the satisfiability tests implicit in our results naturally yield an effective way of checking whether the observed choice behavior of an economic agent is induced by an underlying rationality on the set of alternatives.

Future research on the topic is related to the extension of the current approach to a more general setting. 
In this direction, it is natural to examine the satisfiability problem for semantics characterized by other types of axioms of choice consistency, which are connected to rationalizability issues.
We are currently studying the lifting (and the associated satisfiability problem) of several combinations of axioms, namely, $(\alpha)$ adjoined with $(\gamma)$ and/or $(\rho)$ (see~\cite{CanGiaWat17}). 
The motivation of this analysis is that the satisfaction of these axioms is connected to the (quasi-transitive) rationalizability of a choice. 
More generally, it appears natural to examine the lifting of $(m,n)$\textit{-Ferrers properties} in the sense of \cite{GiaWat14} (see also~\cite{GiaWat17}).
In fact, these properties give rise to additional types of rationalizability -- the so-called $(m,n)$\textit{-rationalizability} -- in which the relation of revealed preference satisfies structural forms of pseudo-transitivity (see~\cite{CanGiaGreWat16}).

\newcommand{\TLQSTR}{\ensuremath{\mbox{$\mathrm{3LQST_{0}^{R}}$}}\xspace}
\newcommand{\QLQSR}{\ensuremath{\mbox{$4\mathrm{LQS}^{R}$}}\xspace}
We also intend to find decidable extensions with choice correspondence terms of the three-sorted fragment of set theory $\TLQSTR$ (see \cite{CN16}) and of the four-sorted fragments $(\QLQSR)^h$, with $h \in \mathbb{N}$ (see \cite{CanNic13a}), which admit a restricted form of quantification over individual and set variables (in the case of $\TLQSTR$), and also over collection variables (in the case of $(\QLQSR)^h$. The resulting decision procedures would allow to reason automatically on very expressive properties in choice theory. 

\section*{Acknowledgements}
Thanks are due to the referees for their useful comments.

%
%
%
%

\bibliographystyle{eptcs}

\end{document}